\documentclass[12pt,a4paper]{article}
\usepackage{amssymb}
\usepackage[onehalfspacing]{setspace}

\newtheorem{theorem}{Theorem}

\newtheorem{claim}[theorem]{Claim}

\newtheorem{definition}[theorem]{Definition}

\newtheorem{lemma}[theorem]{Lemma}

\newenvironment{proof}[1][Proof]{\noindent\textbf{#1.} }{\ \rule{0.5em}{0.5em}}

\input{tcilatex}
\begin{document}

\title{The Query Complexity of Correlated Equilibria\thanks{%
Dedicated to the memory of Lloyd S. Shapley: a giant in the field, a
pioneering and inspiring figure, a supportive teacher and mentor, and a
friend. The combination of game theory with operations research,
combinatorics, probability, and computer science---all present in this
paper---has been a cornerstone of Lloyd Shapley's work. Interestingly, the
edge iso-perimetric inequality (Hart 1976) that we use here came about in
order to solve a problem posed by Lloyd in 1974 in connection with the
Banzhaf value.}~\thanks{%
This version: December 2016. Previous versions: May 2013, September 2013
(Center for Rationality DP-647). Part of this research was carried out at
Microsoft Research, Silicon Valley. We thank Parikshit Gopalan for helpful
discussions leading to the proof of Theorem B, Yakov Babichenko, Kevin
Leyton-Brown, Christos Papadimitriou, Tim Roughgarden, Eva Tardos, and Ricky
Vohra for useful discussions, and the referees and editor for their careful
reading and comments.}}
\author{Sergiu Hart\thanks{%
Institute of Mathematics, Department of Economics, and Center for the Study
of Rationality, Hebrew University of Jerusalem. Research partially supported
by Advanced Investigator Grant 249159 of the European Research Council
(ERC). \emph{e-mail}: \texttt{hart@huji.ac.il} \qquad \emph{web page}: 
\texttt{http://www.ma.huji.ac.il/hart}} \and Noam Nisan\thanks{%
School of Computer Science and Engneering, and Center for the Study of
Rationality, Hebrew University of Jerusalem. Part of this work was supported
by ISF grants 230/10 and 1435/14 of the Israeli Academy of Sciences.} }
\maketitle

\begin{abstract}
We consider the complexity of finding a \emph{correlated equilibrium} of an $%
n$-player game in a model that allows the algorithm to make queries on
players' payoffs at pure strategy profiles. Randomized regret-based dynamics
are known to yield an approximate correlated equilibrium efficiently,
namely, in time that is polynomial in the number of players $n$. Here we
show that\emph{\ \textbf{both} randomization }and\emph{\ approximation} are
necessary: no efficient deterministic algorithm can reach even an
approximate correlated equilibrium, and no efficient randomized algorithm
can reach an exact correlated equilibrium. The results are obtained by
bounding from below the number of payoff queries that are needed.

\end{abstract}

%TCIMACRO{%
%\TeXButton{References Without Numbers}{\def\@biblabel#1{#1\hfill}
%\def\thebibliography#1{\section*{References}
%\addcontentsline{toc}{section}{References}
%\list
%{}{
%\labelwidth 0pt
%\leftmargin 1.8em
%\itemindent -1.8em
%\usecounter{enumi}}
%\def\newblock{\hskip .11em plus .33em minus .07em}
%\sloppy\clubpenalty4000\widowpenalty4000
%\sfcode`\.=1000\relax\def\baselinestretch{1}\large \normalsize}
%\let\endthebibliography=\endlist}}%
%BeginExpansion
\def\@biblabel#1{#1\hfill}
\def\thebibliography#1{\section*{References}
\addcontentsline{toc}{section}{References}
\list
{}{
\labelwidth 0pt
\leftmargin 1.8em
\itemindent -1.8em
\usecounter{enumi}}
\def\newblock{\hskip .11em plus .33em minus .07em}
\sloppy\clubpenalty4000\widowpenalty4000
\sfcode`\.=1000\relax\def\baselinestretch{1}\large \normalsize}
\let\endthebibliography=\endlist%
%EndExpansion

\section{Introduction}

The computational complexity of various notions of equilibrium in games is
of interest in many different models of computation. In the present paper we
focus on the important concept of \emph{correlated equilibrium}, introduced
by Aumann (1974). Perhaps the most striking positive result in this vein is
the surprising power of \textquotedblleft regret-based" algorithms in
finding (approximate) correlated equilibria. These algorithms are obtained
from a large family of natural dynamics that converge to correlated
equilibria in any game, by emulating the computations carried out by the
players in the game; for these dynamics, see, e.g., Foster and Vohra (1997),
Hart and Mas-Colell (2000, 2001), Blum and Mansour (2007), and also the
books of Cesa-Bianchi and Lugosi (2006) and Hart and Mas-Colell (2013). Some
of these algorithms have been shown to converge efficiently (i.e.,
\textquotedblleft quickly"; see Cesa-Bianchi and Lugosi 2003, 2006 and the
Remark at the end of the Introduction). This is in contrast to the fact that
there is no natural dynamic converging to Nash equilibria (Hart and
Mas-Colell 2003, 2006; 2013) and, in fact, no efficient algorithm is
believed to exist either (see Nisan et al. 2007 [Chapter 2]).

Looking at these dynamics from a strictly computational point of view, they
yield algorithms that take the payoff (or utility) functions $%
(u_{1},...,u_{n})$ of $n$ players as input, 
%where each $u_i : S_1 \times \cdots \times S_n \rightarrow [0,1]$
and produce an (approximate) correlated equilibrium 
%$x \in \Delta(S_1 \times \cdots \times S_n)$
as output. Assuming that each player has $m$ pure strategies, % $|S_i|=m$,
the input size is $n\cdot m^{n}$, and yet these algorithms run in time that
is polynomial in $n$ and\footnote{%
As well as in $1/\varepsilon ,$ the inverse of the approximation parameter.} 
$m$---which is sub-linear (even poly-logarithmic) in the input size. Even
though the output size $m^{n}$ (a probability for each $n$-tuple of
strategies) is similar to the input size, these algorithms nevertheless
produce an output whose support is also small: polynomial in $n$ and $m$.
These regret-based algorithms need only \textquotedblleft black
box\textquotedblright\ access to the payoff functions, namely, the
possibility of making a sequence of queries $u_{i}(s_{1},...,s_{n})$ for
pure strategy profiles $(s_{1},...,s_{n})$.

However, the regret-based algorithms have two undesirable aspects: first,
they are \emph{randomized}, and second, they produce only an \emph{%
approximate} equilibrium.\footnote{%
It will follow from our results below that running these algorithms with
high enough precision---so that rounding up an approximate equilibrium
yields an exact equilibrium---would not help, as an exponentially high
precision would be required.} The present paper asks whether these
shortcomings can be fixed. While for many problems with sub-linear
algorithms it is clear that both randomization and approximation are
required, this is not the case here. Usually, the necessity of randomized
approximation is already implied by the verification of the result itself,
also known as \textquotedblleft certificate complexity\textquotedblright\ or
\textquotedblleft non-deterministic complexity.\textquotedblright \footnote{%
Let us look, as an example, at the prototypical sub-linear algorithm of
statistical sampling. For input $x\in \{0,1\}^{n}$, the task is to
compute---perhaps approximately---the fraction of $1$'s in the input, i.e., $%
\sum_{i}x_{i}/n$. A randomized approximation algorithm samples $%
O(\varepsilon ^{-2})$ entries and gets, with high probability, an $%
\varepsilon $-approximation. In this case it is easy to see that
randomization and approximation are both crucial since, even if the answer
were given to the algorithm---say, exactly half of the input bits are $1$%
---verifying that this is so would require querying essentially all inputs,
whether this be done deterministically (even approximately) or exactly (even
randomly).} However, correlated equilibria can be verified in time that is
polynomial in the size of their support, which itself can be polynomial in $%
n $ and $m$, and so the certificate (non-deterministic) complexity of
correlated equilibrium is small.\footnote{%
To verify that a distribution with support of size $k$ is a correlated
equilibrium calls for checking at most $n\cdot m\cdot (m-1)$ linear
inequalities, and these require at most $n\cdot m\cdot k$ payoff queries.}

Another indication that exact or deterministic algorithms may be possible
comes from the known linear programming (LP) based algorithms for correlated
equilibria (Papadimitriou and Roughgarden 2008, Jiang and Leyton-Brown 2011)
that produce, in time that is polynomial in $m$ and $n$, a correlated
equilibrium \emph{exactly and deterministically}. However, this is obtained
in the stronger model where the algorithm may query the payoff black boxes
also at profiles of \emph{mixed} strategies.\footnote{%
For example, the single mixed query of what is player $1$'s expected payoff
when each player plays the uniform mixed strategy (i.e., each pure strategy
has the same probability of $1/m)$ requires $m^{n}$ pure queries.} 
%, i.e. for distributions $x_1 \in \Delta(S_1) ... x_n \in \Delta(S_n)$ get in a single query the value
%$u_i(x_1 ... x_n) = \sum_{s_1 \in S_1 ... s_n \in S_n} x_1(s_1) \cdot x_2(s_2) \cdot \cdots \cdot x_n(s_n) u_i(s_1 ... s_n)$
We note that an even stronger query model appears in the \textquotedblleft
communication complexity" setup, where any function of the payoff matrices
may be queried, e.g., Hart and Mansour (2010).

In recent work, Babichenko and Barman (2013) showed that every \emph{%
deterministic} algorithm that finds an \emph{exact} correlated equilibrium
requires a number of pure queries that is exponential in the number of
players $n$. This raises the question of whether the success of the
regret-based algorithms is due to the power of randomization (which is known
to be critical for achieving low regret), or to the relaxation that allows
approximate equilibria rather than exact ones.

In the present paper we show that it is actually \emph{both} of these, even
when we limit ourselves to bi-strategy games\footnote{%
For bi-strategy games \emph{coarse correlated equilibria} (also known as
\textquotedblleft Hannan equilibria") are equivalent to correlated
equilibria, and so the lower bounds obtained here apply also to the easier
problem of finding coarse correlated equilibria.} (i.e., $m=2$). The
following lower bounds apply in the query model\ (or \textquotedblleft
decision tree model\textquotedblright ), which makes no assumptions on the
algorithm (in particular not restricting its computational power) beyond the
fact that it has black-box\ access to the payoff functions, which provide
the values $u_{i}(s_{1},...,s_{n})$ for an adaptively chosen sequence of
pure strategy profiles\footnote{%
Since we are proving lower bounds, the strong (perhaps unrealistic) model
only strengthens our results. We should note that all the regret-based
algorithms mentioned above are actually effective and not only do they make
few queries but also the rest of the computation requires only polynomial
time.} $(s_{1},...,s_{n})$ and players $i.$ In particular, this model
abstracts away from issues having to do with the way that the input is
accessed.\footnote{%
Input queries presume \textquotedblleft random access memory"; without it,
all algorithms become at least linear in the input size, and thus
exponential in $n.$}

Our results are:\footnote{%
We use the standard notations: $f(n)=O(g(n))$ when there is $c>0$ such that $%
f(n)\leq cg(n)$ for all $n,$ and $f(n)=\Omega (g(n))$ when there is $c>0$
such that $f(n)\geq cg(n)$ for all $n.$ Also, $f(n)=o(g(n))$ if $%
f(n)/g(n)\rightarrow 0$ as $n\rightarrow \infty .$}

\bigskip

\noindent \textbf{Theorem A.} \emph{Every \textbf{deterministic} algorithm
that finds a }$1/2$\emph{-approximate correlated equilibrium in every }$n$%
\emph{-person} \emph{bi-strategy game with payoffs in }$\{0,1\}$\emph{\
requires }$2^{\Omega (n)}$\emph{\ queries in the worst case.}

\bigskip

\noindent (Of course, this holds \emph{a fortiori} for an $\varepsilon $%
-approximate correlated equilibrium for any $0<\varepsilon \leq 1/2.)$

\bigskip

\noindent \textbf{Theorem B.} \emph{Every algorithm (randomized or
deterministic) that finds an \textbf{exact} correlated equilibrium in every }%
$n$\emph{-person bi-strategy game with payoffs specified as }$b$\emph{-bit
integers with }$b=\Omega (n)$\emph{\ incurs an }$2^{\Omega (n)}$\emph{\
expected cost in the worst case.}

\bigskip

In Theorem B the \textquotedblleft cost" includes the number of queries
together with the size of the support of the output produced; see Section 2
(d) for a discussion of this issue. We do not know whether the result
continues to hold for smaller payoffs, such as payoffs with $b=o(n)$ bits,
or payoffs in $\{0,1\}$ (i.e., $b=1);$ see Section 2 (e).

\bigskip

The following table summarizes the results on the number of queries---we
refer to these as \textquotedblleft \emph{query complexity bounds}%
\textquotedblright :\footnote{%
\textquotedblleft Approximate Corr Eq" stands for $\varepsilon $-approximate
correlated equilibrium for fixed (small enough) $\varepsilon >0;$ the number
of strategies $m$ is also fixed $(m=2);$ and the \textquotedblleft
regret-based" algorithms are discussed in the Remark immediately below.}%
\[
\renewcommand{\arraystretch}{1.4}%
\begin{tabular}{c||c|c||}
\cline{2-3}
& \multicolumn{2}{c||}{\textbf{Algorithm}} \\ 
& \textbf{Randomized} & \textbf{Deterministic} \\ \hline\hline
\multicolumn{1}{||c||}{\textbf{Approximate Corr Eq}} & 
\begin{tabular}{c}
$O\left( n\log n\right) $ \\ 
(regret-based)%
\end{tabular}
& 
\begin{tabular}{c}
$2^{%
%TCIMACRO{\TeXButton{\displaystyle}{\displaystyle}}%
%BeginExpansion
\displaystyle%
%EndExpansion
\Omega (n)}$ \\ 
(Theorem A)%
\end{tabular}
\\ \hline
\multicolumn{1}{||c||}{\textbf{Exact Corr Eq}} & 
\begin{tabular}{c}
$2^{%
%TCIMACRO{\TeXButton{\displaystyle}{\displaystyle}}%
%BeginExpansion
\displaystyle%
%EndExpansion
\Omega (n)}$ \\ 
(Theorem B)%
\end{tabular}
& 
\begin{tabular}{c}
$2^{%
%TCIMACRO{\TeXButton{\displaystyle}{\displaystyle}}%
%BeginExpansion
\displaystyle%
%EndExpansion
\Omega (n)}$ \\ 
(B\&B 2013)%
\end{tabular}
\\ \hline\hline
\end{tabular}%
\]

\bigskip

\noindent \textbf{Remark: Regret-based randomized algorithms for computing
approximate correlated equilibria. }Assume that each player has at most $m$
pure strategies. In Cesa-Bianchi and Lugosi (2003; 2006, Remark 7.6 in
Section 7.4) it is shown that by running a regret-based procedure one finds,
with probability at least $1/2,$ an $\varepsilon $-correlated equilibrium in
at most $T=16\ln (2nm)/\varepsilon ^{2}$ steps. Because each player makes no
more than $m$ payoff queries at each step, the total number of queries is $%
\leq nmT.$ After these $T$ steps one checks whether the regrets are all $%
\leq \varepsilon $ (no further queries are needed here, as the regrets are
computed all along) and, if they are not, one starts the procedure afresh.
Because the probability of success is at least $1/2,$ the expected number of
repetitions is $2,$ and so the expected total number of queries is a most $%
2nmT=32nm\ln (2nm)/\varepsilon ^{2},$ which is $O(n\log n)$ for fixed $m$
and $\varepsilon >0.$

\section{Extensions, Variations, and Open Problems}

In this section we discuss a number of relevant issues and open problems.

\bigskip

\noindent \textbf{(a) Query complexity of linear programming}

Since computing a correlated equilibrium (CE)\ in an $n$-person bi-strategy
game is a linear programming (LP) problem with $N=2^{n}$ nonnegative
unknowns (the probabilities of the $2^{n}$ strategy profiles) and $2n+1$
linear constraints (two inequalities per player; in addition, the
probabilities sum up to $1),$ it is appropriate to ask what is the query
complexity of general LP problems of this size. Here, one queries the
coefficients appearing in the various constraints.

Consider a linear programming problem with $N$ nonnegative unknowns and just 
$2$ constraints:%
\begin{equation}
\sum_{j=1}^{N}a_{j}x_{j}\geq 0,\;\;\sum_{j=1}^{N}x_{j}=1,\text{~\ ~}x\geq 0.
\label{eq:LP}
\end{equation}

\begin{claim}
Finding a $1/2$-approximate solution to problem (\ref{eq:LP}) requires $%
\Omega (N)$ queries on the coefficients $(a_{j})_{j=1,...,N}$ in the worst
case.
\end{claim}

\begin{proof}
For every $k=1,...,N$ let (\ref{eq:LP})$_{k}$ be the instance of problem (%
\ref{eq:LP}) with $a_{k}=1$ and $a_{j}=-3$ for all $j\neq k.$ For every $%
\varepsilon \geq 0,$ all the $\varepsilon $-approximate solutions of (\ref%
{eq:LP})$_{k}$ (where the inequality is relaxed to $\sum_{j=1}^{N}a_{j}x_{j}%
\geq -\varepsilon )$ satisfy $x_{k}\geq (3-\varepsilon )/4$ and $\sum_{j\neq
k}x_{j}\leq (1+\varepsilon )/4,$ and so for $\varepsilon \leq 1/2$ we get $%
x_{k}\geq 5/8>3/8\geq x_{j}$ for all $j\neq k.$ Therefore the algorithm must
find $k$ in $\{1,...,N\},$ which requires $\Omega (N)$ queries (whether
deterministic or randomized).
\end{proof}

\bigskip

The query complexity of an LP of size comparable to the correlated
equilibrium LP is thus $\Omega (N)=\Omega (2^{n}),$ i.e., exponential in $n.$
This immediately implies that regret-based algorithms \emph{cannot} be
efficiently translated to general LP problems.

\bigskip

\noindent \textbf{(b) Correlated equilibrium as special linear programming}

As seen in (a) above, the correlated equilibrium LP must have a special
structure that distinguishes it from general LP problems of similar size.
What is that structure, and how does it help to get the fast (i.e.,
polynomial in $n)$ convergence of randomized algorithms to approximate
correlated equilibria?

One feature is that the dual LP of the correlated equilibrium LP decomposes
into $n$ separate LP problems, each one of size $m\times m$ (where $m$ is
the number of strategies of each player). This \textquotedblleft dual
separability" feature lies at the basis of the existence proof of Hart and
Schmeidler (1989), is used in the algorithm of Papadimitriou and Roughgarden
(2008) and Jiang and Leyton-Brown (2011), and translates to
\textquotedblleft uncoupledness" in the world of game dynamics (cf. Hart and
Mas-Colell 2003, 2006; 2013). While this feature distinguishes the
correlated equilibrium LP\ from other LP problems, it does not explain why
it helps in \emph{only }one of the four cases (see the table at the end of
the Introduction). We thus have:

\bigskip

\textbf{Open Problem 1.} Why does the special structure of the correlated
equilibrium\ LP help \emph{only} for randomized algorithms yielding
approximate solutions (where the query complexity is polynomial rather than
exponential in $n$), and not in any of the other cases (where the query
complexity is exponential in $n)$?

\bigskip

\noindent \textbf{(c) Support size of approximate correlated equilibria}

What is the minimal support size that guarantees existence of an $%
\varepsilon $-approximate correlated equilibrium in every $n$-person
bi-strategy game? An $\varepsilon $-correlated equilibrium is just an $%
\varepsilon $-optimal strategy in a two-person zero-sum game where the
opponent has $2n$ strategies (that correspond to the correlated equilibrium
inequalities; this zero-sum game is the \textquotedblleft auxilliary game"
of Hart and Schmeidler 1989). It follows, by using the result of Lipton and
Young (1994), that there always exist $\varepsilon $-correlated equilibria
with uniform support of size\footnote{%
\textquotedblleft Uniform support of size $k$" means that the support
consists of $k$ strategy profiles, not necessarily distinct, each one with
weight $1/k$ (alternatively, the probability weights are all integer
multiples of $1/k).$} $k=\log n/(2\varepsilon ^{2})$.

As for dynamics, it has been shown (Cesa-Bianchi and Lugosi 2003, 2006; see
the Remark at the end of the Introduction) that there are regret-based
procedures that in $T=16\ln n(4n)/\varepsilon ^{2}$ steps reach an $%
\varepsilon $-correlated equilibrium with probability at least $1/2$;
moreover, the resulting $\varepsilon $-correlated equilibrium has uniform
support of size $T$ (it gives equal weight of $1/T$ to each one of the $n$%
-tuples of strategies played in the first $T$ periods). The fact that $T$ is
no more than a constant multiple of $k$ is remarkable, as it implies that
regret-based algorithms yield $\varepsilon $-correlated equilibria with
support that is essentially minimal, and so they converge as fast as
theoretically possible (up to a constant factor). In terms of queries, this
translates to an upper bound of $O(n\log n)$ queries (because each period
every one of the $n$ players makes $2$ queries; cf. Goldberg and Roth 2014).
For improved bounds on the support size, see Babichenko, Barman and Peretz
(2014).

\bigskip

\noindent \textbf{(d) Cost does not include the size of the support of the
output}

As we have seen in the Introduction, if the output of the algorithm is a
distribution that has small support (i.e., polynomial in $n),$ then
computing payoffs and verifying the correlated equilibrium inequalities
requires only polynomially many queries. By contrast, if the support is
large (i.e., exponential in $n)$ these computations require exponentially
many queries (and that is so even if the representation is succinct, e.g.,
the product of uniform mixed strategies). Therefore, our model counts the
size of the support as part of the cost. Interestingly, for Theorem A it
turns out that this issue does not matter (we show this in Section 3.3). But
it may well matter for Theorem B. When the output's support size is \emph{not%
} counted, our proofs show that the number of queries is $2^{\Omega (n)}$ in
the worst case for randomized algorithms that yield correct answers with
probability one (see footnote \ref{ft:2} in Section \ref{s:randomized}).
However, we do not know whether this is so also for randomized algorithms
that are required to yield correct answers with high probability, and
possibly incorrect answers otherwise (we conjecture that it is).

\bigskip

\textbf{Open Problem 2.} Does the result of Theorem B hold also for
randomized algorithms that yield correct answers with high probability
(rather than with probability one) and for which the size of the support of
the output is not counted?

\bigskip

\noindent \textbf{(e) Exact correlated equilibria for games with small
payoffs}

Our proof of Theorem B uses, for the worst case, games whose payoffs range
up to $2^{\Omega (n)};$ we do not know whether this requirement is needed,
and so we have:

\bigskip

\textbf{Open Problem 3.} Does the result of Theorem B hold also for payoffs
with $b=o(n)$ bits, and even for payoffs in $\{0,1\}$ (i.e., $b=1$)?

\bigskip

\noindent \textbf{(f) Query complexity of Nash equilibria}

The lower bounds of Theorems A and B apply also to the harder problem of
finding a Nash equilibrium (since every Nash equilibrium is also a
correlated equilibrium), but it is not difficult to see that, in contrast to
the correlated case, for Nash equilibria these bounds are \textquotedblleft
trivial\textquotedblright\ as they apply also to the verification complexity.%
\footnote{%
Consider $n/2$ pairs of players, where each pair of players is playing their
own matching pennies game. Clearly the unique Nash equilibrium has every
player uniformly randomizing between his two strategies. However, verifying
this equilibrium---even allowing randomized verification---requires looking
at essentially all $2^{n}$ strategy profiles since if an adversary had
changed the utility of a player in any single profile, this would no longer
be an equilibrium. For deterministic verification the adversary could change
the utility at all non-queried profiles so that this would no longer be even
an approximate equilibrium.} However, if we allow both randomization and
approximation then the verification complexity of Nash equilibrium becomes
polynomial in $n$ and $m$ (since we can verify that each player $i$ is
approximately best-responding by sampling from the distributions of the
other players). In earlier versions of the present paper we raised the
following question:

\bigskip

\textbf{Problem 4.} Fix $\varepsilon >0;$ does there exist a randomized
algorithm, with only black-box access to the players' payoff functions, that
finds an $\varepsilon $-Nash equilibrium for every $n$-player $m$-strategy
game whose running time is polynomial in $n$ and $m$?

\bigskip

This problem was recently solved by Babichenko (2014) and Chen, Cheng, and
Tang (2015), who proved that the answer is negative: the query complexity of
approximate Nash equilibria is exponential.

\section{Model and Preliminaries}

\subsection{Correlated Equilibrium}

The setup and notations are standard:

\begin{itemize}
\item \textbf{Game:} We will consider games between $n$ players, where each
player's pure strategy set is $\{0,1\}$. Our players' payoffs are normalized
between $0$ and $1$, and so a game is given by $n$ payoff functions $%
u_{1},...,u_{n},$ where for each $i$ we have $u_{i}:\{0,1\}^{n}\rightarrow
\lbrack 0,1]$.

\item \textbf{Notation:} For a (pure) strategy profile $v\in \{0,1\}^{n}$,
we use two notations: $v^{(i)}\in \{0,1\}^{n}$ denotes the result of
flipping the $i$'th bit of $v$ (i.e., where player $i$ plays $1-v_{i}$ and
all other $j$ play $v_{j}$); and $v^{i\rightarrow b}\in \{0,1\}^{n}$ denotes 
$v$ with the $i$'th bit set to $b$ (i.e., if $v_{i}=b$ then $v^{i\rightarrow
b}=v$, and otherwise $v^{i\rightarrow b}=v^{(i)}$.)

\item \textbf{Regret:} Let $x$ be a probability distribution over the set of
(pure) strategy profiles, i.e., $x:\{0,1\}^{n}\rightarrow \lbrack 0,1]$ with 
$\sum_{v\in \{0,1\}^{n}}x(v)=1$. Take a player $1\leq i\leq n$ and a
possible strategy $b\in \{0,1\}$ for $i$. We say that the \emph{regret of $i$
for not playing $b$} is: 
\[
Regret_{i\rightarrow b}(x)=\sum_{v\in \{0,1\}^{n}}x(v)u_{i}(v^{i\rightarrow
b})-\sum_{v\in \{0,1\}^{n}}x(v)u_{i}(v). 
\]

\item \textbf{Correlated Equilibrium:} For $\varepsilon \geq 0$, we say that 
$x$ is an $\varepsilon $-\emph{correlated equilibrium} ($\varepsilon $-\emph{%
CE}) if for all $i=1,...,n$ and for all $b\in \{0,1\}$ we have $%
Regret_{i\rightarrow b}(x)\leq \varepsilon $. When $\varepsilon =0$ it is a 
\emph{correlated equilibrium} (\emph{CE}).
\end{itemize}

\subsection{The Computational Problem}

The computational problem is as follows.

\bigskip

\noindent \textbf{The (Approximate) Correlated Equilibrium (CE) problem:}

\begin{itemize}
\item \textbf{Input:} The payoff functions $u_{i}:\{0,1\}^{n}\rightarrow
\lbrack 0,1]$ for $i=1,...,n$ and the desired approximation parameter $%
\varepsilon \geq 0$.

\item \textbf{Queries:} We assume only black-box access to the payoffs: a
query is of the form $v\in \{0,1\}^{n},$ and the reply to the query is the $%
n $-tuple of player payoffs: $u_{1}(v),...,u_{n}(v)$.

\item \textbf{Output:} An $\varepsilon $-correlated equilibrium of the game.
The equilibrium is given by listing the probability $x(v)$ for every
strategy profile $v$ in its support.

\item \textbf{Cost: }The cost on a given input $u_{1},...,u_{n}$ and $%
\varepsilon \geq 0$ is the total number of queries made plus the size of the
support of the equilibrium produced. The cost of the algorithm is the
worst-case cost over all $n$-tuples of payoff functions.
\end{itemize}

\subsection{Concise Equilibrium Representations\label{weak}}

Notice that in the definition of the CE problem we counted the size of the
support of the produced equilibrium toward the cost of the algorithm. We
could alternatively talk about the \emph{weak-CE problem} where the
algorithm is allowed to produce an equilibrium with arbitrarily large
support whose size is not counted as part of the cost. Practically, the
algorithm could use some concise representation of the equilibrium, e.g.,
some mixture of product distributions. We show however that this would not
make the problem significantly easier, as any approximate CE algorithm that
makes a small number of queries can be converted to one that also produces a
CE with small support.

\begin{lemma}
\label{concise} Given an algorithm that solves the weak-CE problem with
error $\varepsilon $ in $T$ queries, there exists an algorithm that solves
the (strong) CE problem with error $O(\sqrt{\varepsilon })$ and cost
(including support size) $O(T)$.
\end{lemma}

We first bound the probability that a CE algorithm may assign to un-queried
profiles.

\begin{lemma}
\label{l:1}Consider any $\varepsilon $-CE algorithm that on some input $%
(u_{1},...,u_{n})$ queries a set $Q$ of profiles and outputs an $\varepsilon 
$-CE $x$. Denote $Q^{\prime }=\{v|v\in Q$ or $v^{(1)}\in Q\}$ and $\alpha
=\max_{v\in Q^{\prime }}u_{1}(v)$; then $x$ puts probability at most $%
2(\alpha +\varepsilon )$ outside $Q^{\prime }$, i.e., $\sum_{v\not\in
Q^{\prime }}x(v)\leq 2(\alpha +\varepsilon )$.
\end{lemma}

\begin{proof}
Assume by way of contradiction that this is not the case, and furthermore
assume without loss of generality that at least half of this weight is on $v$%
's with $v_{1}=0$, i.e., $\sum_{v\not\in Q^{\prime }\text{ and }%
v_{1}=0}x(v)>\alpha +\varepsilon $. Now consider changing the input payoffs
in two different ways without changing the queried profiles. The first way
just assigns $u_{1}(v)=0$ for all $v\not\in Q^{\prime }$. In the second way
we put $u_{1}(v)=0$ for all $v\not\in Q^{\prime }$ with $v_{1}=0$ and $%
u_{1}(v)=1$ for $v\not\in Q^{\prime }$ with $v_{1}=1$. Clearly the gap
between $Regret_{1\rightarrow 1}$ for these two cases is exactly $%
\sum_{v\notin Q^{\prime }\text{ and }v_{1}=0}x(v)>\alpha +\varepsilon $.
However, notice that since $\alpha =\max_{v\in Q^{\prime }}u_{1}(v)$, in the
first way, $u_{1}(v)$ is bounded by $\alpha $ for all $v$ and thus $%
|Regret_{1\rightarrow 1}|\leq \alpha $; in the second way we therefore have $%
Regret_{1\rightarrow 1}>\varepsilon $. This contradicts the correctness of
the $\varepsilon $-CE algorithm on the input obtained by the second way of
changing the original input.
\end{proof}

\bigskip

We now complete the proof of the computational equivalence between the two
versions of the approximate CE problem.

\bigskip

\begin{proof}[Proof of Lemma \protect\ref{concise}]
Let us first run the weak-CE algorithm on $\alpha $-scaled payoffs $%
u_{i}^{\prime }(v)=\alpha u_{i}(v)$ to obtain an $\varepsilon $-CE $%
x^{\prime }(v)$ for the $u^{\prime }$'s (with $\alpha $ to be determined
below). By scaling back we have that $x^{\prime }$ is an $(\varepsilon
/\alpha )$-CE for the original $u$'s: $\sum_{v\in \{0,1\}^{n}}x^{\prime
}(v)(u_{i}(v^{i\rightarrow b})-u_{i}(v))\leq \varepsilon /\alpha $. As in
Lemma \ref{l:1}, denote $Q^{\prime }=\{v|v\in Q$ or $v^{(1)}\in Q\},$ where $%
Q$ is the set of queries made by the algorithm. Our output will be $x(v)=0$
for $v\not\in Q^{\prime }$ and $x(v)=\beta x^{\prime }(v)$ for $v\in
Q^{\prime }$, where $\beta =1/\sum_{v\in Q^{\prime }}x^{\prime }(v)$ is the
scaling factor ensuring that $\sum_{v\in \{0,1\}^{n}}x(v)=1$. By Lemma \ref%
{l:1} $\sum_{v\not\in Q^{\prime }}x^{\prime }(v)\leq 2(\alpha +\varepsilon
), $ and so $\beta \leq 1/(1-2(\alpha +\varepsilon ))$ and 
\begin{eqnarray*}
\sum_{v\in \{0,1\}^{n}}x(v)(u_{i}(v^{i\rightarrow b})-u_{i}(v)) &=&\beta
\sum_{v\in Q^{\prime }}x^{\prime }(v)(u_{i}(v^{i\rightarrow b})-u_{i}(v)) \\
&\leq &\beta \sum_{v\in \{0,1\}^{n}}x^{\prime }(v)(u_{i}(v^{i\rightarrow
b})-u_{i}(v)) \\
&&+\beta \sum_{v\notin Q^{\prime }}x^{\prime }(v)|u_{i}(v^{i\rightarrow
b})-u_{i}(v)| \\
&\leq &\frac{1}{1-2(\alpha +\varepsilon )}\left( \frac{\varepsilon }{\alpha }%
+2(\alpha +\varepsilon )\right)
\end{eqnarray*}%
Choosing $\alpha =O(\sqrt{\varepsilon })$ completes the proof.
\end{proof}

\section{Deterministic Algorithms for Approximate Correlated Equilibria}

This section proves that deterministic algorithms require exponentially many
queries to compute even an approximate equilibrium. Notice that due to Lemma %
\ref{concise} the same is also implied for the weak-CE\ problem that allows
arbitrary concise representations of the output. Our result here is:

\bigskip

\noindent \textbf{Theorem A.} \emph{Every \textbf{deterministic} algorithm
that finds a }$1/2$\emph{-approximate correlated equilibrium in every }$n$%
\emph{-person} \emph{bi-strategy game with payoffs in }$\{0,1\}$\emph{\
requires }$2^{\Omega (n)}$\emph{\ queries in the worst case.}

\subsection{The Approximate Sink Problem}

Our lower bound here will be based on analyzing the following combinatorial
problem on the (Boolean) hypercube $\{0,1\}^{n}.$

\bigskip

\noindent \textbf{The \textquotedblleft Approximate Sink\textquotedblright\
(AS) problem:}

\begin{itemize}
\item \textbf{Input:} A labeling of the edges of the hypercube by
directions; i.e., for every edge $(v,v^{(i)})$ we have a weight (or
direction) $R(v,v^{(i)})\in \{-1,1\}$ such that $R(v,v^{(i)})=-R(v^{(i)},v)$%
. We interpret $R(v,v^{(i)})=1$ as the edge going from $v$ to $v^{(i)},$ and 
$R(v^{(i)},v)=-1$ as the edge going from $v^{(i)}$ to $v.$

\item \textbf{Queries:} The algorithm queries a vertex $v\in \{0,1\}^{n}$
and gets the directions $R(v,v^{(i)})$ of all edges adjacent to $v$.

\item \textbf{Output:} The algorithm wins when it queries a vertex $v$ with
in-degree no less than $n/4$, i.e.,\footnote{$\sum_{i}R(v,v^{(i)})$ is the
net outflow through $v,$ i.e., the out-degree $O(v)$ of $v$ minus its
in-degree $I(n);$ because $O(v)+I(v)=n,$ we have $I(v)\geq n/4$ iff $%
\sum_{i}R(v,v^{(i)})\leq n/2.$} $\sum_{i}R(v,v^{(i)})\leq n/2$.

\bigskip
\end{itemize}

We show that exponentially many queries are needed in order to find such a
vertex. This implies Theorem A due to the following simple reduction:

\begin{lemma}
\label{asace} If there exists a deterministic algorithm that in every game
with payoffs in $\{0,1\}$ finds a $1/2$-CE with at most $T$ queries, then
the $AS$ problem can be solved with at most $T$ queries.
\end{lemma}

\begin{proof}
Given an AS instance we build a CE instance with\footnote{%
Edges are thus pointed in the direction of increasing payoff (i.e., positive
regret): from $v$ to $v^{(i)}$ when $u_{i}(v^{(i)})>u_{i}(v),$ and from $%
v^{(i)}$ to $v$ when $u_{i}(v)>u_{i}(v^{(i)}).$} $%
R(v,v^{(i)})=u_{i}(v^{(i)})-u_{i}(v)$ and run the approximate CE algorithm
on it, translating every query the CE algorithm makes to an AS query. For $%
R(v,v^{(i)})=1$ we set $u_{i}(v^{(i)})=1$ and $u_{i}(v)=0$, while for $%
R(v,v^{(i)})=-1$ we set $u_{i}(v^{(i)})=0$ and $u_{i}(v)=1$. Under this
mapping, when the CE algorithm makes a query $(u_{1}(v),...,u_{n}(v))$ it is
immediately translated to the same query $v$ on the original AS instance,
and $R(v,v^{(i)})=1$ means $u_{i}(v)=0$ while $R(v,v^{(i)})=-1$ means $%
u_{i}(v)=1$. When the CE algorithm produces an approximate equilibrium $x$,
we continue by querying all the profiles in the support of $x$, whose number
was, by definition, already counted toward the cost of the CE algorithm.

Now take a $1/2$-equilibrium output by the CE algorithm. Summing up the
inequalities of the CE we get: $\sum_{i,b}Regret_{i\rightarrow
b}=\sum_{v}x(v)\sum_{i}(u_{i}(v^{(i)})-u_{i}(v))\leq n/2$, which implies
that for some $v$ in the support of $x$ we have $\sum_{i}R(v,v^{(i)})=%
\sum_{i}(u_{i}(v^{(i)})-u_{i}(v))\leq n/2$, as needed.
\end{proof}

\subsection{Polite Algorithms}

We now prove the lower bound for the AS problem. The core of our argument is
to show that every relevant algorithm for the AS problem (see Lemma 9) can
be transformed into an algorithm of the following form, without
significantly increasing the number of queries.

\begin{definition}
We call an algorithm for the AS problem \emph{polite} if whenever a vertex $%
v $ is queried, at least $3n/4$ of its neighbors in the hypercube have not
yet been queried.
\end{definition}

It is quite easy to show that polite algorithms cannot solve AS.

\begin{lemma}
No deterministic polite algorithm can solve the AS problem.
\end{lemma}

\begin{proof}
We provide an adversary argument: whenever a vertex is queried, the
adversary answers with all edges that were previously not committed to
pointing out. Since there are at least $3n/4$ such edges, the in-degree is
at most $n/4$, and so this vertex cannot be an answer.
\end{proof}

\subsection{Closure}

To convert an algorithm to a polite form, we need to make sure that vertices
are queried before too many of their neighbors are. We use the following
notion:

\begin{definition}
For a set $V\subseteq \{0,1\}^{n}$ of vertices in the hypercube, we define
its \emph{closure} $V^{\ast }\subseteq \{0,1\}^{n}$ to be the smallest set
containing $V$ such that for every $v\not\in V^{\ast }$ at most $n/8$ of its
neighbors are in $V^{\ast }$.
\end{definition}

The closure is well defined and can be obtained by starting with $V^{\ast
}=V $ and repeatedly adding to $V^{\ast }$ any vertex $v$ that has more than 
$n/8 $ of its neighbors already in $V^{\ast }$. Clearly the order of
additions does not matter since the number of neighbors a vertex has in $%
V^{\ast }$ only increases as other vertices are added to $V^{\ast }$. When
the process stops every $v\not\in V^{\ast }$ has at most $n/8$ of its
neighbors in $V^{\ast }$.

The point is that we will not need to continue this process of adding
vertices for a long time.

\begin{lemma}
If $|V|<2^{n/8-1}$ then $|V^{\ast }|\leq 2|V|$.
\end{lemma}

\begin{proof}
Assume by way of contradiction that $|V^{\ast }|>2|V|,$ and denote by $U$
the set obtained during the process of building $V^{\ast }$ after adding
exactly $|V|$ vertices; thus $|U|=2|V|$. Let us denote by $e(U)$ the number
of directed edges within $U$, i.e., $e(U)=|\{(u,i)|u\in U$ and $u^{(i)}\in
U\}|$. We provide conflicting lower and upper bounds for $e(U)$. For the
lower bound, notice that every vertex that we added during the process adds
at least $n/4$ edges to $e(U)$ ($n/8$ of its own edges as well as the $n/8$
opposite ones), and so $e(U)\geq |U-V|n/4=|V|n/4$. For the upper bound we
use the edge-isoperimetric inequality on the hypercube (Hart 1976), which
implies that for every subset of the hypercube $e(U)\leq |U|\log _{2}|U|$.
Thus we have $|V|n/4\leq |U|\log _{2}|U|=2|V|(\log _{2}|V|+1)$, and so $%
(1+\log _{2}|V|)\geq n/8,$ contradicting the bound on the size of $V$.
\end{proof}

\subsection{A Polite Simulation}

We can now provide our general simulation by polite algorithms, which
completes the proof of the theorem.

\begin{lemma}
Every algorithm that makes at most $T=2^{n/8-1}$ queries can be simulated by
a polite algorithm that makes at most $2T$ queries.
\end{lemma}

\begin{proof}
For $t=1,...,T$ denote by $q_{t}$ the $t$'th query made by the original
algorithm, and let $Q_{t}=\{q_{1},...,q_{t}\}$ be the set of all queries
made until time $t$. Our polite algorithm will simulate query $q_{t}$ by
querying all vertices in $Q_{t}^{\ast }$, i.e., completing the closure
implied by adding $q_{t}$. Notice that $Q_{t}^{\ast }=(Q_{t-1}\cup
\{q_{t}\})^{\ast }=(Q_{t-1}^{\ast }\cup \{q_{t}\})^{\ast }$. The difficulty
is that we need to add the vertices in $Q_{t}^{\ast }-Q_{t-1}^{\ast }$ in a
way that maintains politeness, i.e., such that each vertex is added before $%
n/4$ of its neighbors are.

To see that this is possible let us look at the vertices in $%
N_{t}=Q_{t}^{\ast }-Q_{t-1}^{\ast }$. First, the previous lemma implies that 
$|N_{t}|\leq |Q_{t}|=t$. By the edge-isoperimetric inequality applied to $%
N_{t}$ we have $e(N_{t})\leq t\log _{2}t$, and so some vertex $v\in N_{t}$
has at most $\log _{2}t<n/8$ neighbors in $N_{t};$ this will be the last
vertex our polite algorithm will query in this stage. Similarly, from the
remaining elements $N^{\prime }=N_{t}\backslash \{v\}$ there is also a
vertex $v^{\prime }$ with at most $\log _{2}(t-1)<n/8$ neighbors in $%
N^{\prime }$, and this vertex will be asked just before $v$. We continue so
until we exhaust $N_{t}$. Now we claim that this order maintains politeness:
since, by definition, every vertex in $N_{t}$ has fewer than $n/8$ neighbors
in $Q_{t-1}^{\ast }$, when we add the fewer than $n/8$ neighbors from $N_{t}$
that appeared \emph{before} it in the ordering of $N_{t}$, we still get
fewer than $n/4$ neighbors preceding it. Finally, notice that the simulating
algorithm queries, by Lemma 8, at most $2T$ vertices in $Q_{T}^{\ast },$ and
so its running time is as required.
\end{proof}

\section{Randomized Algorithms for Exact Correlated Equilibria\label%
{s:randomized}}

This section provides the lower bound for randomized algorithms. Recall that
randomized algorithms can in fact compute an \emph{approximate} CE with
polynomially many queries (using regret-based procedures). This section
proves that they cannot compute an \emph{exact} CE.

First let us formally define a randomized algorithm. A \emph{randomized
algorithm} is just a probability distribution over deterministic algorithms.%
\footnote{%
Which is the same as a \textquotedblleft behavioral" algorithm that makes
randomizations all along (cf. Kuhn's mixed vs. behavioral strategies in
games of perfect recall---which our algorithms clearly have, as we impose no
restrictions such as finite automata).} For every input, this random choice
of the algorithm results in the output being a random variable. We say that
a randomized algorithm solves a search problem (like our problem of finding
an equilibrium) if for every input, the probability that the output is a
correct solution is at least\footnote{%
With the complementary probability the output may be incorrect (and so this
is not a \textquotedblleft zero-error" algorithm).} $1/2$. The \emph{cost}
of a randomized algorithm on a given input is the expected cost made over
the random choice of the algorithm, and the cost of a randomized algorithm
is its cost for the worst-case input. So our theorem for this section is:

\bigskip

\noindent \textbf{Theorem B.} \emph{Every algorithm (randomized or
deterministic) that finds an \textbf{exact} correlated equilibrium in every }%
$n$\emph{-person bi-strategy game with payoffs specified as }$b$\emph{-bit
integers with }$b=\Omega (n)$\emph{\ requires a }$2^{\Omega (n)}$\emph{\
expected cost in the worst case.}

\bigskip

This theorem applies even to randomized algorithms that produce a CE with
any non-negligible probability. It also applies to the weak-CE version of
the problem defined in Section \ref{weak}, but only for zero-error
algorithms.\footnote{\label{ft:2}This is because Lemma \ref{concise} holds
also for zero-error randomized algorithms as it can be applied to each
deterministic algorithm in the support. (Our proof does not imply the
extension to the weak-CE case for general randomized algorithms, even though
we believe that the theorem itself does extend.)} Finally, it can be seen
from the proofs below that it also applies to $\varepsilon $-CE, for $%
\varepsilon $ that is exponentially small in $n$.

\subsection{The Non-Positive Vertex Problem}

Similarly to the deterministic case, we here reduce the correlated
equilibrium problem to the following combinatorial problem on the hypercube.
It is essentially a weighted version of the Approximate Sink problem, with a
stricter bound on the output quality.

\bigskip

\noindent \textbf{The \textquotedblleft Non-Positive
Vertex\textquotedblright\ (NPV) problem:}

\begin{itemize}
\item \textbf{Input:} A labeling of the directed edges of the hypercube by
integers where the convention is that $R(u,v)=-R(v,u)$.

\item \textbf{Queries:} A query is a vertex $v$ in the hypercube. The answer
to this query is the tuple of labels on all adjacent edges: $R(v,v^{(i)})$
for $i=1,...,n$.

\item \textbf{Output:} The algorithm must output a vertex $v$ in the
hypercube with total non-positive weight, i.e., $\sum_{i}R(v,v^{(i)})\leq 0$.
\end{itemize}

\bigskip

As in Lemma \ref{asace}, proving a randomized lower bound for the NPV
problem implies a similar bound for the CE problem due to the following
reduction.

\begin{lemma}
\label{l:npv}The number of queries required for a randomized NPV algorithm
to solve the NPV problem is at most the number of queries required for a
randomized algorithm to solve the CE problem.
\end{lemma}

\begin{proof}
We convert a CE algorithm to an NPV one. Let $m=max_{u,v}R(u,v)$. Given an
NPV instance we build a CE instance ensuring that $%
u_{i}(v^{(i)})-u_{i}(v)=R(v,v^{(i)})/m$: for positive $R(v,v^{(i)})$ we set $%
u_{i}(v)=0$ and $u_{i}(v^{(i)})=R(v,v^{i)})/m$ , while for negative $%
R(v,v^{(i)})$ we set $u_{i}(v)=-R(v,v^{(i)})/m$ and $u_{i}(v^{(i)})=0$.
Under this mapping, when the CE algorithm makes a query $v$ it is
immediately translated to the same query $v$ of the NPV black box, and the
answer from the NPV black box directly provides the answer to the CE query.

Now take an equilibrium $x$ output by the CE algorithm. Summing up the
inequalities of the CE we get: $\sum_{i,b}Regret_{i\rightarrow
b}=\sum_{v}x(v)\sum_{i}(u_{i}(v^{(i)})-u_{i}(v))\leq 0$, which implies that
for some $v$ in the support of $x$ we have $\sum_{i}R(v,v^{(i)})=%
\sum_{i}(u_{i}(v^{(i)})-u_{i}(v))\leq 0$, as needed to provide an answer to
the NPV problem.
\end{proof}

\bigskip

We continue proving the lower bound for randomized algorithms that solve the
NPV problem by exhibiting a distribution over NPV instances such that every 
\emph{deterministic algorithm} requires exponentially many queries in order
to succeed on a non-negligible fraction of inputs drawn according to this
distribution; we appeal here to the so-called \textquotedblleft Yao (1977)
Principle," an instance of von Neumann's Minimax Theorem. The lower bound
for randomized CE algorithms follows, thus completing the proof of Theorem B.

\subsection{A Path Construction}

We build hard instances of the NPV problem from paths in the hypercube. Let $%
(v_{0},v_{1},...,v_{L})$ be a (not necessarily simple) path in the
hypercube; i.e., for each $0\leq j<L$, the vertex $v_{j+1}$ is obtained from 
$v_{j}$ by flipping a single random bit. The NPV instance we build from this
path essentially gives weight $-j$ to the edge $(v_{j-1},v_{j})$, with
weights added over the possible multiple times the path goes through a
single edge. This way every time the path passes a vertex $v$ at step $j$,
the incoming edge gets weight $-j$ while the outgoing edge gets weight $j+1$%
, adding one to the total net weight going out of $v$. Formally:

\begin{definition}
Let $(v_{0},v_{1},...,v_{L})$ be a (not necessarily simple) path in the
hypercube. The path induces the following labeling of the hypercube: $%
R(u,v)=\sum_{\{j|u=v_{j-1},v=v_{j}\}}j-\sum_{\{j|u=v_{j},v=v_{j-1}\}}j$.
\end{definition}

\begin{lemma}
For each $v\neq v_{L}$ we have $\sum_{i}R(v,v^{(i)})=|\{j|v_{j}=v\}|$. For $%
v_{L}$ we have $\sum_{i}R(v_{L},v_{L}^{(i)})=|\{j|v_{j}=v_{L}\}|-L$.
\end{lemma}

\begin{proof}
Except for the last vertex, $v_{L}$, whenever the path goes into $v$ at step 
$j$ and exits it in step $j+1$, the total values of $R$'s going out of this
vertex increases by 1 ($-j$ incoming and $j+1$ outgoing). As for $v_{L},$ we
need to subtract $L$ due to the fact that the last edge goes in (and there
is no edge that goes out).
\end{proof}

\bigskip

This means that if the path covers the whole hypercube then the only
non-positive vertex in it is the end of the path.

We now define a random distribution over paths that cover the whole
hypercube and whose final end point is random. We start with some (fixed)
Hamiltonian path in the hypercube. From that point on we continue with a
random walk of length $L=n\cdot 2^{n/3}$. Why would it be hard for an
algorithm to find the end of the path? We show that an algorithm must
essentially follow the path query by query. Otherwise it is looking for a
needle of length $L=n\cdot 2^{n/3}$ that is randomly hidden in a haystack of
size $2^{n}$. But probing places that are already known to be in the random
part of the path only allows the algorithm to advance sequentially over it,
thus requiring exponential time to reach the end.

\begin{lemma}
\label{l:2}Any deterministic algorithm that runs in time $T<2^{n/3}/n$ is
able to solve the NPV problem on at most a fraction of $O(2^{-n/3})$ of
inputs drawn according to this distribution.
\end{lemma}

Since all the information in our path-based instances of the NPV problem are
determined by the path, we can imagine that the queries directly ask for
this path information. Since the only non-positive vertex in these instances
is the end of the path, our algorithm really needs to find it. Not only is
this hard to do, but it is hard even to find any vertex that is near the
tail of the path. This is so since, on the one hand, the tail is a tiny
fraction of the hypercube and so it can't be found \textquotedblleft at
random,\textquotedblright\ and, on the other hand, the only possible
\textquotedblleft deliberate\textquotedblright\ way to find it is to follow
the path step by step, which takes exponential time. The reason that no
\textquotedblleft shortcuts\textquotedblright\ are possible when following
the path is that the random walk in the hypercube mixes rapidly, and so one
gets no information about how the path continues beyond the very near
vicinity. To formalize this line of reasoning we introduce a variant of the
problem that explicitly provides the algorithm with any information that we
think it may get a handle on. Specifically, we tell the algorithm everything
about the path except for its tail, and, furthermore, give it an additional $%
n^{2}$ vertices at the beginning of this tail every step. Once this is given
to the algorithm, we are able to show that the algorithm can never learn
anything new.

\subsection{\textbf{The \textquotedblleft Hit The Path\textquotedblright\
(HTP) problem}}

We formally define the HTP problem.

\bigskip

\noindent \textbf{The \textquotedblleft Hit The Path\textquotedblright\
(HTP) problem:}

\begin{itemize}
\item \textbf{Input:} A random path $(v_{0},...,v_{L})$ in the hypercube of
length $L=n2^{n/3}$, starting from a \emph{revealed} vertex $v_{0}$.

\item \textbf{Queries:} At each step $t=1,...,T$:

\begin{enumerate}
\item The algorithm may query a vertex $q_{t}$ of the hypercube, depending
on the revealed information so far (which we will soon see is exactly the
sequence $(v_{0},...,v_{(t-1)n^{2}})$).

\item The next $n^{2}$ vertices of the path, i.e., $%
v_{(t-1)n^{2}+1},...,v_{tn^{2}}$, are revealed (independently of the query).
\end{enumerate}

\item \textbf{Output: }The algorithm wins at time $t$ if $q_{t}$ is on the 
\emph{non-previously revealed} part of the path, i.e., $q_{t}\in
\{v_{tn^{2}+1},...,v_{L}\}$.
\end{itemize}

\begin{lemma}
\label{l:new}Every algorithm for the HTP problem with at most $T$ queries
wins with probability at most $(1+o(1))n2^{-2n/3}T.$
\end{lemma}

\begin{proof}
Fix a (deterministic) algorithm for the HTP problem that makes at most $T$
queries. If it wins, then for some step $1\leq t\leq T$ it won (for the
first time) by finding an unrevealed vertex $v_{j}$ for $tn^{2}<j\leq L$. We
bound this probability (over the random choice of the path) for a fixed $t$
and $j$, and then use the union bound to obtain an upper bound on the
probability that the algorithm wins. Now let us look at the $t$'th query $%
q_{t}$ made by the algorithm. If none of the previous queries won, then the
only information the algorithm had when making this query was the revealed
part of the path, i.e., $(v_{0},v_{1},...,v_{(t-1)n^{2}})$, and so the query
is just a function of these: $q_{t}=q_{t}(v_{0},v_{1},...,v_{(t-1)n^{2}})$.
What is the probability that for some function on these inputs we have $%
q_{t}(v_{0},v_{1},...,v_{(t-1)n^{2}})=v_{j}$? Note that our construction of
random paths means that $v_{j}$ is obtained by taking a random walk of
length $j-(t-1)n^{2}\geq n^{2}$ from vertex $v_{(t-1)n^{2}}$. Now comes the
crucial observation: as the mixing time of the hypercube is known to be $%
O(n\log n)<n^{2}$ (cf. Diaconis et al. 1990), this means that a random walk
of length $l\geq n^{2}$ ends at an almost uniformly random vertex of the
hypercube. Thus for any fixed $(v_{0},v_{1},...,v_{(t-1)n^{2}})$, we have
that $v_{j}$ is almost uniformly distributed over the hypercube and so $%
Pr[q_{t}(v_{0},v_{1},...,v_{(t-1)n^{2}})=v_{j}]=(1+o(1))2^{-n}$. Multiplying
this quantity by $T<2^{n/3}/n$ (for all possible values of $t$) and then by $%
2^{n/3}n$ (for all possible values of $j$), we get the required upper bound
for the probability of winning.
\end{proof}

\bigskip

\begin{proof}[Proof of Lemma \protect\ref{l:2}]
First, note that Lemma \ref{l:new} implies in particular that an
HTP-algorithm with $T<2^{n/3}/n$ queries can win with probability at most $%
O(2^{-n/3})$.

Second, recalling the discussion immediately following the statement of
Lemma \ref{l:2}, suppose that we have an algorithm that succeeds in solving
the NPV problem on a larger fraction of inputs drawn according to this
distribution; we use it to win instances of this HTP problem with at least
the same probability. Whenever the NPV algorithm makes a query to $v$ we
make the same query in the HTP case. If we win, then we are done. Otherwise,
we know that the non-revealed part of the path does not pass through $v$ and
so the reply to the NPV query is completely determined by the revealed part
of the path, which we already have and can use for the reply. If the NPV
algorithm succeeds then it must have found the last vertex on the path (the
only non-positive one), which is on the path and is revealed only after $%
L/n^{2}=2^{n/3}/n>T$ queries, and so is still unrevealed and thus our HTP
algorithm wins too.
\end{proof}

\bigskip

Combining Lemmas \ref{l:npv} and \ref{l:2} proves Theorem B.

\end{document}